\documentclass{article}
\usepackage{graphicx}
\usepackage{multirow}

\usepackage[T1]{fontenc}
\usepackage[sc]{mathpazo}
\usepackage{amsmath}
\usepackage{amssymb}
\usepackage{enumerate}
\usepackage{amsthm}% theorems and proofs
\usepackage{amsfonts,mathrsfs}
\usepackage{hyperref}
\hypersetup{pdfpagemode=UseNone}

\newcommand{\tinyspace}{\mspace{1mu}}

\newcommand{\abs}[1]{\left\lvert\tinyspace #1 \tinyspace\right\rvert}

\newcommand{\norm}[1]{\left\lVert\tinyspace #1 \tinyspace\right\rVert}

\newcommand{\setft}[1]{\mathrm{#1}}

\newcommand{\density}[1]{\setft{D}\left(#1\right)}
\newcommand{\unitary}[1]{\setft{U}\left(#1\right)}

\newcommand{\pd}[1]{\setft{Pd}\left(#1\right)}

\newcommand{\supp}{{\operatorname{supp}}}

\def\dif{\mathrm{d}}

\def\complex{\mathbb{C}}
\def\real{\mathbb{R}}

\def\I{\mathbb{1}}

\newenvironment{mylist}[1]{\begin{list}{}{
    \setlength{\leftmargin}{#1}
    \setlength{\rightmargin}{0mm}
    \setlength{\labelsep}{2mm}
    \setlength{\labelwidth}{8mm}
    \setlength{\itemsep}{0mm}}}
    {\end{list}}

\def\ot{\otimes}

\newcommand{\out}[2]{| #1\rangle\langle #2 |}

\newcommand{\Innerm}[3]{\left\langle #1 \left| #2 \right| #3 \right\rangle}
\newcommand{\defeq}{\stackrel{\smash{\textnormal{\tiny def}}}{=}}

\newcommand{\pa}[1]{(#1)}
\newcommand{\Pa}[1]{\left(#1\right)}

\newcommand{\Br}[1]{\left[#1\right]}
\newcommand{\set}[1]{\{#1\}}
\newcommand{\Set}[1]{\left\{#1\right\}}

\newcommand{\bra}[1]{\langle#1|}

\newcommand{\ket}[1]{|#1\rangle}

\DeclareMathOperator{\trace}{Tr}
\newcommand{\ptr}[2]{\trace_{#1}\pa{#2}}
\newcommand{\Ptr}[2]{\trace_{#1}\Pa{#2}}

\newcommand{\Tr}[1]{\Ptr{}{#1}}

\def\cH{\mathcal{H}}
\def\cM{\mathcal{M}}

\def\rS{\mathrm{S}}

\def\A{\textsf{A}}

\newtheorem{thrm}{Theorem}[section]
\newtheorem{lem}[thrm]{Lemma}
\newtheorem{prop}[thrm]{Proposition}
\newtheorem{cor}[thrm]{Corollary}
\theoremstyle{definition}

\newtheorem{remark}[thrm]{Remark}

\numberwithin{equation}{section}

\newcounter{questionnumber}

\begin{document}

%==================================================================================================================================%
\title{\Large A strengthened monotonicity inequality of quantum relative entropy: \\A unifying approach via R\'{e}nyi relative entropy}
%==================================================================================================================================%

\author{Lin Zhang\footnote{E-mail: godyalin@163.com;
linyz@zju.edu.cn}\\
  {\small\it Institute of Mathematics, Hangzhou Dianzi University, Hangzhou 310018, PR~China}}
\date{}
\maketitle
\maketitle \mbox{}\hrule\mbox\\
\begin{abstract}

We derive a strengthened monotonicity inequality for quantum
relative entropy by employing properties of $\alpha$-R\'{e}nyi
relative entropy. We develop a unifying treatment towards the
improvement of some quantum entropy inequalities. In particular, an
emphasis is put on a lower bound of quantum conditional mutual
information (QCMI) as it gives a Pinsker-like lower bound for the
QCMI. We also give some improved entropy inequalities based on
R\'{e}nyi relative entropy. The inequalities obtained, thus, extends
some well-known ones. We also obtain a condition under which a
tripartite operator becomes a Markov state. As a by-product we
provide some trace inequalities of
operators, which are of independent interest.
\end{abstract}
\maketitle \mbox{}\hrule\mbox\\

%=============================================================================%
\section{Introduction}
%=============================================================================%

To begin with, let us fix some notations. Let $\cH_d$ be a
$d$-dimensional complex Hilbert space. A \emph{quantum state} $\rho$
on $\cH_d$ is a positive semi-definite operator of trace one and in
particular, the operator $\rho = \out{\psi}{\psi}$ is said to be a
\emph{pure state} for each unit vector $\ket{\psi} \in \cH_d$. The
set of all quantum states on $\cH_d$ is denoted by
$\density{\cH_d}$. For each quantum state $\rho\in\density{\cH_d}$,
its von Neumann entropy is defined by $\rS(\rho) := -
\Tr{\rho\log\rho}$. Here and in remaining parts, all the logarithms
are taken with respect to the natural base $e$. The \emph{relative
entropy} of two mixed states $\rho$ and $\sigma$ is defined by
\begin{eqnarray}
\rS(\rho||\sigma) := \left\{\begin{array}{ll}
                             \Tr{\rho(\log\rho -
\log\sigma)}, & \text{if}\ \supp(\rho) \subseteq
\supp(\sigma), \\
                             +\infty, & \text{otherwise}.
                           \end{array}
\right.
\end{eqnarray}
Here $\supp(\rho)$ ($\supp(\sigma)$) means the support set of $\rho$
($\sigma$). A \emph{quantum channel} $\Phi$ over $\cH_d$ is defined
as a trace-preserving completely positive linear map over the set
$\density{\cH_d}$. It follows that there exist linear operators
$\set{K_\mu}_\mu$ on $\cH_d$ such that $\sum_\mu K^\dagger_\mu K_\mu
= \I$ and $\Phi = \sum_\mu \mathrm{Ad}_{K_\mu}$, where
$\mathrm{Ad}_{K_\mu}(X):=K_\mu XK^\dagger_\mu$, that is, for each
quantum state $\rho$, we have the Kraus representation $\Phi(\rho) =
\sum_\mu K_\mu \rho K^\dagger_\mu$. A well-known property of quantum
relative entropy is its monotonicity under generic quantum channels.
That is,
\begin{eqnarray}
\rS(\rho||\sigma) \geqslant \rS(\Phi(\rho)||\Phi(\sigma)).
\end{eqnarray}
The condition of equality in the above equation is an interesting
and important subject. An extremely important result in quantum
information theory, namely, the saturation of monotonicity
inequality of relative entropy under a generic quantum channel, is
provided by Petz \cite{Petz1988} and we restate it below.
\begin{prop}[Petz, \cite{Hiai2011,Petz1988}]
Let $\rho,\sigma\in\density{\cH_d}$ and $\Phi$ be a quantum channel
defined over $\cH_d$. If $\supp(\rho)\subseteq\supp(\sigma)$, then
\begin{eqnarray}
\rS(\rho||\sigma) = \rS(\Phi(\rho)||\Phi(\sigma))\quad\text{if and
only if}\quad \Phi^*_\sigma\circ\Phi(\rho) = \rho,
\end{eqnarray}
where $\Phi^*_\sigma =
\mathrm{Ad}_{\sigma^{1/2}}\circ\Phi^*\circ\mathrm{Ad}_{\Phi(\sigma)^{-1/2}}$,
and $\Phi^*$ is the dual of $\Phi$ with respect to Hilbert-Schmidt
inner product over the operator space on $\cH_d$, i.e.
$\Tr{\Phi^*(X)Y}=\Tr{X\Phi(Y)}$ for all operators $X,Y$ on $\cH_d$.
\end{prop}

The well-known strong subadditivity (SSA) inequality of quantum
entropy, obtained by Lieb and Ruskai in \cite{Lieb1973},
\begin{eqnarray}\label{eq:SSA-1}
\rS(\rho_{ABC}) +
\rS(\rho_B) \leqslant \rS(\rho_{AB}) + \rS(\rho_{BC}),
\end{eqnarray}
is a ubiquitous result in quantum information theory. It is known
that SSA is equivalent to the monotonicity inequality of quantum
relative entropy. Based on SSA, a new concept--conditional mutual
information--is proposed by mimicking classical one. It measures the
correlations of two quantum systems relative to a third: Given a
tripartite state $\rho_{ABC}\in\density{\cH_{ABC}}$, where
$\cH_{ABC}:=\cH_A\ot\cH_B\ot\cH_C$, it is defined as
\begin{eqnarray}
I(A:C|B)_\rho := \rS(\rho_{AB})+ \rS(\rho_{BC}) - \rS(\rho_{ABC}) -
\rS(\rho_B).
\end{eqnarray}
Clearly conditional mutual information is nonnegative by SSA. Thus,
getting a lower bound on conditional mutual information is
equivalent to the tightening of SSA and is an important line of
research. Hence, characterization of vanishing conditional mutual
information is a first step to this problem.

Ruskai is the first one to discuss the equality condition of SSA,
i.e. vanishing conditional mutual information. By analyzing the
equality condition of Golden-Thompson inequality, she obtained the
following characterization \cite{Ruskai2002}:
\begin{eqnarray}
I(A:C|B)_\rho = 0 \Longleftrightarrow \log\rho_{ABC} + \log\rho_B =
\log\rho_{AB} + \log\rho_{BC}.
\end{eqnarray}
Throughout the present paper, we have suppressed implicit tensor
products with the identity by conventions. For example,
$\log\rho_{AB}$ means $(\log\rho_{AB})\ot\I_C$, where
$\rho_{AB}=\ptr{C}{\rho_{ABC}}$ is the reduced state (or density
operator) of the system $AB$.

Later on, using the relative modular approach established by Araki,
Petz gave another characterization of the equality condition of SSA
\cite{Petz2003}:
\begin{eqnarray}
I(A:C|B)_\rho = 0 \Longleftrightarrow
\rho^{\mathrm{i}t}_{ABC}\rho^{-\mathrm{i}t}_{BC} =
\rho^{\mathrm{i}t}_{AB} \rho^{-\mathrm{i}t}_B\quad(\forall
t\in\real),
\end{eqnarray}
where $\mathrm{i} = \sqrt{-1}$ is the imaginary unit. From this, we
see that $I(A:C|B)_\rho > 0$ if and only if $
\rho^{\mathrm{i}t}_{ABC} \neq \rho^{\mathrm{i}t}_{AB}
\rho^{-\mathrm{i}t}_B\rho^{\mathrm{i}t}_{BC}$ for all $t\in\real$.
Therefore, comparing both $I(A:C|B)_\rho$ and
$\norm{\rho^{\mathrm{i}t}_{ABC} - \rho^{\mathrm{i}t}_{AB}
\rho^{-\mathrm{i}t}_B\rho^{\mathrm{i}t}_{BC}}$, where $\norm{*}$ is
a metric over the unitary group, is mathematically interesting. For
instance, they have a property in common: If
$\rho,\sigma\in\density{\cH_d}$, then $\norm{\rho -
\sigma}_1\leqslant 2$ and $\norm{\rho^{\mathrm{i}t} -
\sigma^{\mathrm{i}t}}_\infty\leqslant 2$, where $\norm{*}_\infty$ is
the spectral norm of an operator.

Hayden \emph{et al.} in \cite{Hayden2004} have shown that
$I(A:C|B)_\rho =0$ if and only if the following conditions hold:
\begin{enumerate}[(i)]
\item $\cH_B = \bigoplus_k \cH_{b^L_k} \ot \cH_{b^R_k}$,
\item $\rho_{ABC} = \bigoplus_k  p_k \rho_{Ab^L_k} \ot \rho_{b^R_kC}$, where $\rho_{Ab^L_k}\in\density{\cH_A \ot \cH_{b^L_k}}, \rho_{b^R_kC} \in \density{\cH_{b^R_k} \ot\cH_C}$
for each index $k$; and $\set{p_k}$ is a probability distribution.
\end{enumerate}

In order to get rid of the above-known difficult computation such as
logarithm and complex exponential power of states, Zhang
\cite{Zhang2013} gave another new characterization of vanishing
conditional mutual information. Specifically, define
\begin{eqnarray*}
M &:=&
(\rho^{1/2}_{AB}\ot\I_C)(\I_A\ot\rho^{-1/2}_B\ot\I_C)(\I_A\ot\rho^{1/2}_{BC})\equiv
\rho^{1/2}_{AB}\rho^{-1/2}_B\rho^{1/2}_{BC}.
\end{eqnarray*}
Then the following conditions are equivalent:
\begin{enumerate}[(i)]
\item The vanishing of conditional mutual information, i.e. $I(A:C|B)_\rho = 0$.
\item $\rho_{ABC} = MM^\dagger = \rho^{1/2}_{AB}\rho^{-1/2}_B\rho_{BC}\rho^{-1/2}_B \rho^{1/2}_{AB}$.
\item $\rho_{ABC} = M^\dagger M = \rho^{1/2}_{BC}\rho^{-1/2}_B\rho_{AB}\rho^{-1/2}_B \rho^{1/2}_{BC}$.
\end{enumerate}
With these characterizations of vanishing conditional mutual
information, one starts to make an attempt to get a lower bound on
the QCMI. In \cite{Brandao2011}, Brand\~{a}o \emph{et al.} first
obtained the following lower bound for $I(A:C|B)_\rho$:
\begin{eqnarray}\label{eq:Fernando}
I(A:C|B)_\rho \geqslant \frac18
\min_{\sigma_{AC}\in\mathbb{SEP}}\norm{\rho_{AC} -
\sigma_{AC}}^2_{1-\mathbb{LOCC}},
\end{eqnarray}
where
$$
\norm{\rho_{AC} - \sigma_{AC}}^2_{1-\mathbb{LOCC}} \defeq
\sup_{\cM\in1-\mathbb{LOCC}}\norm{\cM(\rho_{AC}) -
\cM(\sigma_{AC})}_1.
$$
Here $\mathbb{SEP}$ means the set of all separable states over the
bipartite cut $A:C$; 1-$\mathbb{LOCC}$ means the one-way LOCC
measurement. They used many advanced mathematical techniques to get
this result in their paper. Based on this result, they cracked a
\emph{long-standing} open problem in quantum information theory: the
squashed entanglement is \emph{faithful}. Later, Li and Winter in
\cite{Li2014} gave another approach to study the same problem and
have improved the lower bound for $I(A:C|B)_\rho$:
\begin{eqnarray}\label{eq:keli}
I(A:C|B)_\rho \geqslant \frac12
\min_{\sigma_{AC}\in\mathbb{SEP}}\norm{\rho_{AC} -
\sigma_{AC}}^2_{1-\mathbb{LOCC}}.
\end{eqnarray}

A different approach is taken by Ibinson \emph{et al.} in
\cite{Ibinson2008}. They studied the robustness of \emph{quantum
Markov chains}, i.e. the perturbation to the states with vanishing
conditional mutual information. They found that the quantum Markov
chains are not robust because, even if the conditional mutual
information is small, the original tripartite state can deviate a
lot from Markov chains.

Several breakthroughs about the investigation of bounding the small
conditional mutual information are made, respectively, by Fawzi and
Renner \cite{FR}, Wilde \emph{et. al} \cite{BLW,BSW,DW,SBW},
Brand\~{a}o \cite{BHOS}, Li and Winter \cite{Li2014,LW}, and, Zhang
and Wu \cite{Lin}. In this paper, we give a unifying treatment for
some entropy inequalities and improvement of monotonicity inequality
of relative entropy under unital quantum channels by employing
quantum $\alpha$-R\'{e}nyi relative entropy \cite{Mosonyi}. Once we
get one of our main results
(Theorem~\ref{th:stronger-monotonicity}), we can simply derive all
improved versions of some quantum entropy inequalities
\cite{Carlen}. Note that our method is different from that in
\cite{Carlen}, and is much simpler compared to the one in
\cite{Lin}.

The paper is organized as follows. The definition and properties of
quantum R\'{e}nyi relative entropy are given in
Section~\ref{sect:Renyi}. Section~\ref{sect:Main-result} deals with
the main results and their consequences. In the
Section~\ref{sect:entropy-inequalities}, we summarize a series of
strengthened entropy inequalities. The discussion and concluding
remarks are presented in Section~\ref{sect:Remarks}. Some questions
are left open for the future research.

%=================================================================%
\section{Quantum $\alpha$-R\'{e}nyi relative entropy}\label{sect:Renyi}
%=================================================================%

The \emph{quantum $\alpha$-R\'{e}nyi relative entropy} is defined as
follows \cite{Mosonyi}:
\begin{eqnarray}
\rS_\alpha(\rho||\sigma):=
\frac1{\alpha-1}\log\Tr{\rho^\alpha\sigma^{1-\alpha}},
\end{eqnarray}
where $\rho,\sigma\in\density{\cH_d}$, and a parameter
$\alpha\in(0,1)$. Two important properties of $\alpha$-R\'{e}nyi
relative entropy used in this paper are listed below: it holds that
\begin{enumerate}[(i)]
\item $\rS(\rho||\sigma) =
\lim_{\alpha\to1^-}\rS_\alpha(\rho||\sigma)$, thus we denote
$\rS(\rho||\sigma) = \rS_{1^-}(\rho||\sigma)$;
\item $\alpha\mapsto \rS_\alpha(\rho||\sigma)$ is monotonically
increasing on (0,1).
\end{enumerate}
Hence, if $\alpha\geqslant\frac12$, then
$\rS_{1/2}(\rho||\sigma)\leqslant \rS_\alpha(\rho||\sigma)$. Taking
the limit for $\alpha\to1^-$ on the right hand side of the
inequality, we have the following important result \footnote{This
method to get the inequality is pointed out to the author by M.
Wilde.}:
\begin{eqnarray}\label{eq:renyi-divergence}
\rS(\rho||\sigma) \geqslant -2\log\Tr{\sqrt{\rho}\sqrt{\sigma}}
\end{eqnarray}
for two states $\rho$ and $\sigma$. The same inequality is obtained
by Carlen and Lieb \cite{Carlen} using Peierls-Bogoliubov inequality
and Golden-Thompson inequality. Later, in our further
investigations, we find that this inequality seems to improve some
of the entropy inequalities obtained recently. Compared with
Pinsker's bound for relative entropy
\begin{eqnarray}
\rS(\rho||\sigma)\geqslant\frac12\norm{\rho-\sigma}^2_1,
\end{eqnarray}
the lower bound in Eq.~\eqref{eq:renyi-divergence} for the relative
entropy is very useful in our present paper. Note that there is an
identity which will be used in our treatment:
$$
\rS(\rho||\mu\sigma) = \rS(\rho||\sigma) - \log\mu,~~\forall \mu>0.
$$

%=================================================================%
\section{Main results}\label{sect:Main-result}
%=================================================================%

In this section, we prove our main theorem. We take a unifying
method to treat improvement of some entropy inequalities. The proof
strategy followed here is completely different from that used by
Carlen and Lieb \cite{Carlen}. In the following proposition, we
present a very simple and natural proof following our unifying
approach.
\begin{prop}\label{prop:stronger-monotonicity}
For two states $\rho,\sigma\in\density{\cH_d}$ and a quantum channel
$\Phi$ over $\cH_d$, we have
\begin{eqnarray*}
&&\rS(\rho||\sigma) - \rS(\Phi(\rho)||\Phi(\sigma)) \\
&&\geqslant -2\log\Tr{\sqrt{\rho}\sqrt{\exp\Br{\log\sigma +
\Phi^*(\log\Phi(\rho)) - \Phi^*(\log\Phi(\sigma))}}}.
\end{eqnarray*}
\end{prop}

\begin{proof}
Define a state as follows:
$$
\omega = \lambda^{-1}\exp(\log\sigma +\Phi^*(\log\Phi(\rho) -
\log\Phi(\sigma))),
$$
with $\lambda:=\Tr{\exp(\log\sigma +\Phi^*(\log\Phi(\rho) -
\log\Phi(\sigma)))}>0$. Multiplying by $\lambda$ on both sides of
above equation and then taking logarithm, we get
\begin{eqnarray}\label{Eq:one}
\log\sigma +\Phi^*(\log\Phi(\rho) - \log\Phi(\sigma)) =
\log(\lambda\omega).
\end{eqnarray}
Now consider
\begin{eqnarray*}
&&\rS(\rho||\sigma) - \rS(\Phi(\rho)||\Phi(\sigma)) \\
&&= \Tr{\rho(\log\rho - \log\sigma)} - \Tr{\Phi(\rho)(\log\Phi(\rho)
-
\log\Phi(\sigma))}\\
&&=\Tr{\rho(\log\rho - \log\sigma)} - \Tr{\rho\Phi^*(\log\Phi(\rho)
- \log\Phi(\sigma))}\\
&&= \Tr{\rho\Pa{\log\rho - [\log\sigma +\Phi^*(\log\Phi(\rho) -
\log\Phi(\sigma))]}}.
\end{eqnarray*}
Using Eq. \eqref{Eq:one}, we have
\begin{align*}
\rS(\rho||\sigma) - \rS(\Phi(\rho)||\Phi(\sigma)) &=\rS(\rho||\lambda\omega)\\
&= \rS(\rho||\omega)-\log\lambda\\
&\geqslant -2\log\Tr{\sqrt{\rho}\sqrt{\omega}} - \log\lambda.
\end{align*}
Since
$$
\sqrt{\omega} = \lambda^{-1/2}\sqrt{\exp(\log\sigma
+\Phi^*(\log\Phi(\rho) - \log\Phi(\sigma)))},
$$
it follows that
$$
-2\log\Tr{\sqrt{\rho}\sqrt{\omega}} - \log\lambda =
-2\log\Tr{\sqrt{\rho}\sqrt{\exp(\log\sigma +\Phi^*(\log\Phi(\rho) -
\log\Phi(\sigma)))}}.
$$
This completes the proof.
\end{proof}

\begin{remark}\label{rem:1}
For any given two positive semi-definite matrices $M$ and $N$, it
holds \cite{Lin} that
\begin{eqnarray}\label{eq:matrix-inequality}
\norm{\sqrt{M} - \sqrt{N}}^2_2 \leqslant\norm{M-N}_1\leqslant
\norm{\sqrt{M} - \sqrt{N}}_2 \norm{\sqrt{M} + \sqrt{N}}_2,
\end{eqnarray}
where $\norm{X}_p:= \Pa{\Tr{\abs{X}^p}}^{1/p}$ is Schatten $p$-norm
for positive integers $p$, and $\abs{X}=\sqrt{X^\dagger X}$. Indeed,
the proof of Eq.~\eqref{eq:matrix-inequality} uses the well-known
inequality in matrix analysis, i.e. \emph{Audenaert's inequality}
\cite{Audenaert}:
\begin{eqnarray}
\Tr{M^tN^{1-t}}\geqslant\frac12\Tr{M+N - \abs{M-N}}
\end{eqnarray}
for all $t\in[0,1]$ and positive matrices $M,N$. If both the traces
of $M$ and $N$ are no more than one, i.e. $\Tr{M},
\Tr{N}\leqslant1$, then we see from the proof of \cite[Theorem
2.1]{Lin} that
$$
\Tr{\sqrt{M}\sqrt{N}} \leqslant 1 - \frac12\norm{\sqrt{M} -
\sqrt{N}}^2_2.
$$
Furthermore,
\begin{eqnarray}
-2\log\Tr{\sqrt{M}\sqrt{N}} \geqslant -2\log\Pa{1 -
\frac12\norm{\sqrt{M} - \sqrt{N}}^2_2} \geqslant \norm{\sqrt{M} -
\sqrt{N}}^2_2,
\end{eqnarray}
where we used the fact that $-\log(1-t)\geqslant t$ for $t\leqslant
1$.
\end{remark}

\begin{prop}[Lieb, \cite{Lieb}]\label{th:Lieb-concavity}
For a fixed Hermitian matrix $H\in M_d(\complex)$, the following map
\begin{eqnarray}
X\mapsto \Tr{e^{H+\log X}}
\end{eqnarray}
is concave over the set $\pd{\complex^d}$ of all positive definite
matrices of order $d$.
\end{prop}

\begin{prop}\label{prop:mono-under-ptrace}
For two arbitrary bipartite states
$\rho_{AB},\sigma_{AB}\in\density{\cH_{AB}}$ with
$\cH_{AB}=\cH_A\ot\cH_B$, it holds that
\begin{eqnarray}
\rS(\rho_{AB}||\sigma_{AB}) - \rS(\rho_A||\sigma_A)
&\geqslant&-2\log\Tr{\sqrt{\rho_{AB}} \sqrt{\exp(\log\sigma_{AB} -
\log\sigma_A + \log\rho_A)}}\\
&\geqslant& \norm{\sqrt{\rho_{AB}} - \sqrt{\exp(\log\sigma_{AB} -
\log\sigma_A + \log\rho_A)}}^2_2.
\end{eqnarray}
In particular, $\rS(\rho_{AB}||\sigma_{AB}) = \rS(\rho_A||\sigma_A)$
if and only if $\log\rho_{AB} - \log\rho_A = \log\sigma_{AB} -
\log\sigma_A$.
\end{prop}

\begin{proof}
In Proposition~\ref{prop:stronger-monotonicity}, letting
$\rho=\rho_{AB}$, $\sigma=\sigma_{AB}$, and the quantum channel
$\Phi=\trace_B$ (a partial trace over system $B$), we obtain the
first inequality. The second inequality follows from
Remark~\ref{rem:1} due to the fact that $\Tr{\exp(\log\sigma_{AB} -
\log\sigma_A + \log\rho_A)}\leqslant1$. Indeed, let $H=\log\rho_A -
\log\sigma_A$ and $X=\sigma_{AB}$ in
Proposition~\ref{th:Lieb-concavity}, it follows that
\begin{eqnarray*}
&&\Tr{\exp(\log\rho_A -
\log\sigma_A+\log\sigma_{AB})}\\
&&=\int_{\unitary{d_B}}\Tr{\exp\Pa{\log\rho_A\ot \I_B -
\log\sigma_A\ot \I_B + \log
(U_B\sigma_{AB}U^\dagger_B)}}dU_B\\
&&\leqslant \Tr{\exp\Pa{\log\rho_A\ot \I_B - \log\sigma_A\ot\I_B +
\log
\Br{\int_{\unitary{d_B}}U_B\sigma_{AB}U^\dagger_B\dif U_B}}}\\
&&=\Tr{\exp\Pa{\log\rho_A\ot \I_B -
\log\sigma_A\ot\I_B + \log \Br{\sigma_A\ot \I_B/d_B}}}\\
&&=\Tr{\exp\Pa{\log\rho_A\ot \I_B - \log\sigma_A\ot\I_B +
\log\sigma_A\ot\I_B +
\I_A\ot\log(\I_B/d_B)}}\\
&&=\Tr{\rho_A\ot\I_B}/d_B =1.
\end{eqnarray*}
This concludes the proof.
\end{proof}

The above inequality is firstly derived by Carlen and Lieb as one of
their main results in \cite{Carlen}. The following result is a
direct consequence of it. In addition, we present here another
approach to get it.
\begin{cor}\label{th:newbound}
For an arbitrary
tripartite state $\rho_{ABC}$, we have that
\begin{eqnarray}
I(A:C|B)_\rho&\geqslant& -2\log\Tr{\sqrt{\rho_{ABC}}
\sqrt{\exp(\log\rho_{AB} -
\log\rho_B+\log\rho_{BC})}}\\
&\geqslant& \norm{\sqrt{\rho_{ABC}} - \sqrt{\exp(\log\rho_{AB} -
\log\rho_B+\log\rho_{BC})}}^2_2.\label{eq:lower-bound:2-norm}
\end{eqnarray}
In particular, the conditional mutual information vanishes if and
only if $\log\rho_{ABC} + \log\rho_B = \log\rho_{AB} +
\log\rho_{BC}$.
\end{cor}

\begin{proof}
Note that the quantum conditional mutual information $I(A:C|B)_\rho$
can be rewritten as follows:
$$
I(A:C|B)_\rho = \rS(\rho_{ABC}||\omega_{ABC}) - \log\lambda,
$$
where $\lambda\omega_{ABC} = \exp(\log\rho_{AB} -
\log\rho_B+\log\rho_{BC})$ and $\lambda = \Tr{\exp(\log\rho_{AB} -
\log\rho_B+\log\rho_{BC})}$. By using \eqref{eq:renyi-divergence},
we get
\begin{eqnarray*}
I(A:C|B)_\rho &\geqslant&
-2\log\Tr{\sqrt{\rho_{ABC}}\sqrt{\omega_{ABC}}}- \log\lambda\\
&=& -2\log\Tr{\sqrt{\rho_{ABC}} \sqrt{\exp(\log\rho_{AB} -
\log\rho_B+\log\rho_{BC})}}.
\end{eqnarray*}
This is the first inequality. The second approach to the first
inequality is by using Theorem~\ref{th:stronger-monotonicity}. By
letting $\Phi = \trace_A$, $\rho=\rho_{ABC}$ and $\sigma =
\rho_{AB}\ot\rho_C$, it follows that
$$
I(A:C|B)_\rho = \rS(\rho||\sigma) - \rS(\Phi(\rho)||\Phi(\sigma)).
$$
By employing Proposition~\ref{prop:stronger-monotonicity}, we have
$$
I(A:C|B)_\rho \geqslant
-2\log\Tr{\sqrt{\rho_{ABC}}\sqrt{\exp(\log\rho_{AB} -
\log\rho_B+\log\rho_{BC})}}.
$$
The second inequality follows directly from Remark~\ref{rem:1} due
to the fact \cite{Ruskai2002} that
\begin{eqnarray}\label{eq:less1}
\Tr{\exp(\log\rho_{AB} - \log\rho_B+\log\rho_{BC})}\leqslant1.
\end{eqnarray}
Here we give another proof of this inequality. Indeed, let
$H=\log\rho_{AB} - \log\rho_B$ and $X=\rho_{BC}$ in
Proposition~\ref{th:Lieb-concavity}, it follows that
\begin{eqnarray*}
&&\Tr{\exp(\log\rho_{AB} -
\log\rho_B+\log\rho_{BC})}\\
&&=\int_{\unitary{d_C}}\Tr{\exp\Pa{\log\rho_{AB}\ot \I_C -
\I_A\ot\log\rho_B\ot\I_C+\I_A\ot\log
(U_C\rho_{BC}U^\dagger_C)}}\dif U_C\\
&&\leqslant \Tr{\exp\Pa{\log\rho_{AB}\ot \I_C -
\I_A\ot\log\rho_B\ot\I_C+\I_A\ot\log
\Br{\int_{\unitary{d_C}}U_C\rho_{BC}U^\dagger_C\dif U_C}}}\\
&&=\Tr{\exp\Pa{\log\rho_{AB}\ot \I_C -
\I_A\ot\log\rho_B\ot\I_C+\I_A\ot\log \Br{\rho_B\ot \I_C/d_C}}}\\
&&=\Tr{\exp\Pa{\log\rho_{AB}\ot \I_C -
\I_A\ot\log\rho_B\ot\I_C+\I_A\ot\log\rho_B\ot\I_C +
\I_A\ot\I_B\ot\log(\I_C/d_C)}}\\
&&=\Tr{\rho_{AB}\ot\I_C}/d_C =1.
\end{eqnarray*}
Now if the conditional mutual information vanishes, then
$$
\norm{\sqrt{\rho_{ABC}} - \sqrt{\exp\Pa{\log\rho_{AB} +
\log\rho_{BC} -\log\rho_B}}}_2=0,
$$
that is, $\sqrt{\rho_{ABC}} = \sqrt{\exp\Pa{\log\rho_{AB} +
\log\rho_{BC} -\log\rho_B}}$, which is equivalent to the following:
$$
\rho_{ABC} = \exp\Pa{\log\rho_{AB} + \log\rho_{BC} -\log\rho_B}.
$$
By taking logarithm over both sides, it is seen that $\log\rho_{ABC}
= \log\rho_{AB} + \log\rho_{BC} -\log\rho_B$, a well-known equality
condition of strong subadditivity obtained by Ruskai in
\cite{Ruskai2002}. This completes the proof.
\end{proof}

\begin{remark}\label{remark-1}
Note that our technique used in the proof of the inequality
\eqref{eq:less1} implies a more general result: if
$\rho_{ABC},\sigma_{ABC},\tau_{ABC}$ are tripartite states on
$\cH_{ABC}$ satisfying the condition that $\rho_B = \sigma_B$ or
$\sigma_B=\tau_B$, then
\begin{eqnarray}
\Tr{\exp(\log\rho_{AB} - \log\sigma_B+\log\tau_{BC})}\leqslant1.
\end{eqnarray}
\end{remark}

Now we will give an important lemma. In fact, it is based on a
famous Lieb's concavity result, i.e.
Proposition~\ref{th:Lieb-concavity}.
\begin{lem}\label{lem:A1}
For given two states $\rho,\sigma\in\density{\cH_d}$ and a unital
quantum channel $\Phi$ defined over $\cH_d$, i.e. $\Phi(\I_d)=\I_d$
and $\Phi^*(\I_d)=\I_d$, we have
\begin{eqnarray}\label{eq:datta-wilde}
\Tr{\exp\Pa{\log\sigma + \Phi^*(\log\Phi(\rho)) -
\Phi^*(\log\Phi(\sigma))}}\leqslant 1.
\end{eqnarray}
\end{lem}

\begin{proof}
By Stinespring's dilation theorem, a given quantum channel $\Phi$
can be realized as
$$\Phi(X)=\Ptr{B}{U(X\ot \widehat\I_B)U^\dagger}$$
for some unitary $U$ and completely mixed state
$\widehat\I_B:=\I_B/d_B$ in an auxiliary Hilbert space $\cH_B$. The
dual of $\Phi$ is given by $\Phi^*(Y) = \Ptr{B}{U^\dagger
(Y\ot\widehat\I_B)U}$. Denote
$\rho_{AB}:=U(\rho\ot\widehat\I_B)U^\dagger$ and
$\sigma_{AB}:=U(\sigma\ot\widehat\I_B)U^\dagger$. Then
$\rho_A=\ptr{B}{\rho_{AB}}=\Phi(\rho)$ and
$\sigma_A=\ptr{B}{\sigma_{AB}}=\Phi(\sigma)$. Thus by the techniques
in the proof of Proposition~\ref{prop:mono-under-ptrace},
\begin{eqnarray}
\Tr{\exp(\log\rho_A\ot\I_B - \log\sigma_A\ot\I_B +
\log\sigma_{AB})}\leqslant 1.
\end{eqnarray}
That is,
\begin{eqnarray*}
1&\geqslant&\Tr{\exp\Pa{\log\Phi(\rho)\ot\I_B -
\log\Phi(\sigma)\ot\I_B + \log\Br{U(\sigma\ot\widehat\I_B)U^\dagger}}}\\
&=&\Tr{\exp\Pa{U^\dagger\Br{\log\Phi(\rho)\ot\I_B -
\log\Phi(\sigma)\ot\I_B}U + \log(\sigma\ot\widehat\I_B)}}\\
&=&\int_{\unitary{d_B}}\Tr{\exp\Pa{U^\dagger\Br{\log\Phi(\rho)\ot\I_B
- \log\Phi(\sigma)\ot\I_B}U +
\log(\sigma\ot\widehat\I_B)}}\dif V_B\\
&=& \int_{\unitary{d_B}}\Tr{\exp\Pa{(\I\ot
V_B)U^\dagger\Br{\log\Phi(\rho)\ot\I_B -
\log\Phi(\sigma)\ot\I_B}U(\I\ot V_B)^\dagger +
\log(\sigma\ot\widehat\I_B)}}\dif V_B\\
&\geqslant&\Tr{\exp\Pa{\int_{\unitary{d_B}}(\I\ot
V_B)U^\dagger\Br{\log\Phi(\rho)\ot\I_B -
\log\Phi(\sigma)\ot\I_B}U(\I\ot V_B)^\dagger \dif V_B +
\log(\sigma\ot\widehat\I_B)}}\\
&=&\Tr{\exp\Pa{\Ptr{B}{U^\dagger\Br{(\log\Phi(\rho) -
\log\Phi(\sigma))\ot\I_B}U}\ot\widehat\I_B + \log(\sigma\ot\widehat\I_B)}}\\
&=&\Tr{\exp\Pa{\Ptr{B}{U^\dagger\Br{(\log\Phi(\rho) -
\log\Phi(\sigma))\ot\widehat\I_B}U}\ot\I_B +
\log(\sigma\ot\widehat\I_B)}},
\end{eqnarray*}
where we used the following facts: $X\mapsto \Tr{\exp(X)}$ is a
convex functional over $\pd{\complex^n}$ and \cite{LZ}:
\begin{eqnarray}
\int_{\unitary{d_B}} (\I\ot V_B) \Br{U(X\ot\widehat\I_B)U^\dagger}
(\I\ot V_B)^\dagger \dif V_B = \Phi(X)\ot \widehat\I_B.
\end{eqnarray}
This indicates that
\begin{eqnarray*}
1&\geqslant&\Tr{\exp\Pa{\Phi^*(\log\Phi(\rho) -
\log\Phi(\sigma))\ot\I_B + \log(\sigma\ot\widehat\I_B)}}\\
&=&\Tr{\exp\Pa{\Phi^*(\log\Phi(\rho) - \log\Phi(\sigma))\ot\I_B +
\log\sigma\ot\I_B + \I_A\ot\log(\widehat\I_B)}}\\
&=&\Tr{\exp\Pa{[\Phi^*(\log\Phi(\rho) -
\log\Phi(\sigma))+\log\sigma]\ot\I_B + \I_A\ot\log(\widehat\I_B)}}\\
&=&\Tr{\exp\Pa{[\Phi^*(\log\Phi(\rho) -
\log\Phi(\sigma))+\log\sigma]} \ot \exp(\log(\widehat\I_B))}\\
&=&\Tr{\exp\Pa{[\Phi^*(\log\Phi(\rho) -
\log\Phi(\sigma))+\log\sigma]}} \Tr{\widehat\I_B}\\
&=&\Tr{\exp\Pa{[\Phi^*(\log\Phi(\rho) -
\log\Phi(\sigma))+\log\sigma]}},
\end{eqnarray*}
where we used the fact that $\exp(X\ot\I+\I\ot Y) =
\exp(X)\ot\exp(Y)$ for positive definite operators $X$ and $Y$.
\end{proof}

Note that Eq.~\eqref{eq:datta-wilde} is also obtained very recently
by Datta and Wilde \cite{DW}. The approach used here is very simple
and completely different from the one used by them. Now, we may
present our main result---the strengthened monotonicity inequality
of relative entropy---which is described as follows:
\begin{thrm}\label{th:stronger-monotonicity}
For any states $\rho,\sigma\in\density{\cH_d}$, and $\Phi$ a unital
quantum channel over $\cH_d$, we have
\begin{eqnarray}
\rS(\rho||\sigma) - \rS(\Phi(\rho)||\Phi(\sigma))\geqslant
\norm{\sqrt{\rho} - \sqrt{\exp\Pa{\log\sigma +
\Phi^*(\log\Phi(\rho)) - \Phi^*(\log\Phi(\sigma))}}}^2_2.
\end{eqnarray}
\end{thrm}

\begin{proof}
The proof follows immediately from Lemma~\ref{lem:A1} and
Remark~\ref{rem:1}.
\end{proof}

\begin{cor}
With the above notations, we have the following inequalities:
\begin{enumerate}[(i)]
\item Strengthened monotonicity inequality of relative entropy under a unital quantum channel:
\begin{eqnarray} \rS(\rho||\sigma) -
\rS(\Phi(\rho)||\Phi(\sigma))\geqslant \frac14\norm{\rho -
\exp\Pa{\log\sigma + \Phi^*(\log\Phi(\rho)) -
\Phi^*(\log\Phi(\sigma))}}^2_1
\end{eqnarray}
\item Strengthened monotonicity inequality of relative entropy under partial trace:
\begin{eqnarray}
\rS(\rho_{AB}||\sigma_{AB}) - \rS(\rho_A||\sigma_A) \geqslant
\frac14\norm{\rho_{AB} - \exp(\log\sigma_{AB} - \log\sigma_A +
\log\rho_A)}^2_1.
\end{eqnarray}
\item Strengthened subadditivity inequality of quantum entropy:
\begin{eqnarray}\label{eq:conditional-lower-bound}
I(A:C|B)_\rho \geqslant\frac14 \norm{\rho_{ABC} -
\exp\Pa{\log\rho_{AB} + \log\rho_{BC} -\log\rho_B}}^2_1.
\end{eqnarray}
\end{enumerate}
\end{cor}

%===============================================================%
\section{On some entropy inequalities}\label{sect:entropy-inequalities}
%===============================================================%

From our results from the previous sections and the following
result, we can derive many strengthened entropy inequalities.
\begin{prop}\label{prop:univ-result}
For a state $\rho\in\density{\cH_d}$ and a subnormalized state
$\sigma$ on $\cH_d$ (i.e. $\Tr{\sigma}\leqslant1$), it holds that
\begin{eqnarray}\label{eq:sub-normalized-state}
\rS(\rho||\sigma)&\geqslant& -2\log\Tr{\sqrt{\rho}\sqrt{\sigma}}\\
&\geqslant&\norm{\sqrt{\rho}- \sqrt{\sigma}}^2_2\\
&\geqslant& \frac14\norm{\rho-\sigma}^2_1.
\end{eqnarray}
In particular, $\rS(\rho||\sigma)=0$ if and only if $\rho=\sigma$.
\end{prop}
Clearly, all we need to do is to rewrite a related quantity as a
relative entropy with the second argument being a subnormalized
state. Then, Proposition~\ref{prop:univ-result} is applied to get
the desired inequality. From \cite{BSW}, we see that
\begin{eqnarray*}
&&\rS(\rho_{ABC}||\exp(\log\sigma_{AB} + \log\tau_{BC} - \log\omega_B))\\
&&= I(A:C|B)_\rho + \rS(\rho_{AB}||\sigma_{AB}) +
\rS(\rho_{BC}||\tau_{BC}) - \rS(\rho_B||\omega_B),
\end{eqnarray*}
where $\rho_{ABC}\in\density{\cH_{ABC}}$, $\sigma_{AC}
\in\density{\cH_{AC}},\tau_{BC}\in\density{\cH_{BC}}$, and
$\omega_C\in\density{\cH_C}$. This identity leads to the following
result:
\begin{eqnarray*}
&&\rS(\rho_{ABC}||\exp(\log\sigma_{AB} + \log\sigma_{BC} - \log\sigma_B))\\
&&= I(A:C|B)_\rho + \rS(\rho_{AB}||\sigma_{AB}) +
\rS(\rho_{BC}||\sigma_{BC}) - \rS(\rho_B||\sigma_B),
\end{eqnarray*}
where $\rho_{ABC}, \sigma_{ABC}\in\density{\cH_{ABC}}$. Using
monotonicity inequality of relative entropy, we have
$$
\rS(\rho_{AB}||\sigma_{AB}) \geqslant
\rS(\rho_B||\sigma_B)~~\text{and}~~\rS(\rho_{BC}||\sigma_{BC})
\geqslant \rS(\rho_B||\sigma_B).
$$
This yields that
$$
\frac12\Br{\rS(\rho_{AB}||\sigma_{AB}) +
\rS(\rho_{BC}||\sigma_{BC})} \geqslant \rS(\rho_B||\sigma_B).
$$
Therefore, we obtain the following result:
\begin{prop}\label{th:super-strong-additivity}
It holds that
\begin{eqnarray}
&&\rS(\rho_{ABC}||\exp(\log\sigma_{AB} + \log\sigma_{BC} - \log\sigma_B))\notag\\
&&\geqslant I(A:C|B)_\rho + \frac12\rS(\rho_{AB}||\sigma_{AB}) +
\frac12\rS(\rho_{BC}||\sigma_{BC}),
\end{eqnarray}
where $\rho_{ABC}, \sigma_{ABC}\in\density{\cH_{ABC}}$. In
particular, $\rS(\rho_{ABC}||\exp(\log\rho_{AB} + \log\rho_{BC} -
\log\rho_B)) \geqslant 0$, i.e. $I(A:C|B)_\rho\geqslant0$, the
strong subadditivity inequality. Moreover,
$\rS(\rho_{ABC}||\exp(\log\sigma_{AB} + \log\sigma_{BC} -
\log\sigma_B)) = 0$ if and only if $\rho_{ABC} =
\exp(\log\sigma_{AB} + \log\sigma_{BC} - \log\sigma_B)$.
\end{prop}
If $\rS(\rho_{ABC}||\exp(\log\sigma_{AB} + \log\sigma_{BC} -
\log\sigma_B)) = 0$, then using
Proposition~\ref{th:super-strong-additivity}, we have
\begin{eqnarray}
\begin{cases}
I(A:C|B)_\rho &= 0;\\
\rS(\rho_{AB}||\sigma_{AB}) &=0;\\
\rS(\rho_{BC}||\sigma_{BC}) &=0.
\end{cases}
\end{eqnarray}
This leads to the following:
\begin{eqnarray}
\rho_{AB} = \sigma_{AB},~~\rho_{BC} = \sigma_{BC}.
\end{eqnarray}
Thus, $\rho_B=\sigma_B$, which indicates that
$$
\exp(\log\rho_{AB} + \log\rho_{BC} - \log\rho_B) =
\exp(\log\sigma_{AB} + \log\sigma_{BC} - \log\sigma_B).
$$
Note that $I(A:C|B)_\rho = 0$ if and only if $\exp(\log\rho_{AB} +
\log\rho_{BC} - \log\rho_B) = \rho_{ABC}$. Therefore,
$\exp(\log\sigma_{AB} + \log\sigma_{BC} - \log\sigma_B) =
\rho_{ABC}$. From the above-mentioned process, it follows that
$$
\rS(\rho_{ABC}||\exp(\log\sigma_{AB} + \log\sigma_{BC} -
\log\sigma_B))=0 \Longrightarrow \rho_{ABC} = \exp(\log\sigma_{AB} +
\log\sigma_{BC} - \log\sigma_B).
$$
We know that, for any state $\sigma_{ABC}\in\density{\cH_{ABC}}$,
$$
\Tr{\exp(\log\sigma_{AB} + \log\sigma_{BC} - \log\sigma_B)}
\leqslant 1.
$$
But what will happen if $\Tr{\exp(\log\sigma_{AB} + \log\sigma_{BC}
- \log\sigma_B)} = 1$? In order to answer this question, we form an
operator for any state $\sigma_{ABC}\in\density{\cH_{ABC}}$, namely,
$$
\exp(\log\sigma_{AB} + \log\sigma_{BC} - \log\sigma_B).
$$
If $\exp(\log\sigma_{AB} + \log\sigma_{BC} - \log\sigma_B)$ is a
legitimate state, denoted by $\rho_{ABC}$, then
$$
\rho_{AB} = \ptr{C}{\exp(\log\sigma_{AB} + \log\sigma_{BC} -
\log\sigma_B)},~\rho_{BC} = \ptr{A}{\exp(\log\sigma_{AB} +
\log\sigma_{BC} - \log\sigma_B)},
$$
and $\rho_B = \ptr{AC}{\exp(\log\sigma_{AB} + \log\sigma_{BC} -
\log\sigma_B)}$. Furthermore, $\rS(\rho_{ABC}||\exp(\log\sigma_{AB}
+ \log\sigma_{BC} - \log\sigma_B))=0$. Thus, $I(A:C|B)_\rho = 0$,
i.e. $\exp(\log\sigma_{AB} + \log\sigma_{BC} - \log\sigma_B)$ is a
Markov state.
\begin{prop}
Given a state $\rho_{ABC}$, we form an operator $\exp(\log\rho_{AB}
+ \log\rho_{BC} - \log\rho_B)$. If
$$
\Tr{\exp(\log\rho_{AB} + \log\rho_{BC} - \log\rho_B)}=1,
$$
then the following statements are true:
\begin{enumerate}[(i)]
\item $\exp(\log\rho_{AB} + \log\rho_{BC} - \log\rho_B) =
\rho^{1/2}_{AB}\rho^{-1/2}_B\rho_{BC}\rho^{-1/2}_B\rho^{1/2}_{AB}$;
\item $\exp(\log\rho_{AB} + \log\rho_{BC} - \log\rho_B) =
\rho^{1/2}_{BC}\rho^{-1/2}_B\rho_{AB}\rho^{-1/2}_B\rho^{1/2}_{BC}$.
\end{enumerate}
Therefore, $\exp(\log\rho_{AB} + \log\rho_{BC} - \log\rho_B)$ must
be a Markov state.
\end{prop}
From the above result, we see that if a state $\rho_{ABC}$ can be
expressed by the form of $\exp(\log\sigma_{AB} + \log\sigma_{BC} -
\log\sigma_B)$ for another state $\sigma_{ABC}$, then $\rho_{ABC}$
must ba a Markov state. A question naturally arises: Which states
$\rho_{ABC}$ are such that $\exp(\log\rho_{AB} + \log\rho_{BC} -
\log\rho_B)$ is a Markov state? It would be interesting to figure
out the structure of the following set:
\begin{eqnarray}
\Set{\rho_{ABC}\in\density{\cH_{ABC}}: \Tr{\exp(\log\rho_{AB} +
\log\rho_{BC} - \log\rho_B)}=1}.
\end{eqnarray}

\begin{thrm}
For any tripartite states $\rho_{ABC},
\sigma_{ABC},\tau_{ABC},\omega_{ABC}\in\density{\cH_{ABC}}$, if
$\sigma_B=\tau_B$ or $\tau_B=\omega_B$, then
\begin{eqnarray}
&&\rS(\rho_{ABC}||\exp(\log\sigma_{AB} - \log\tau_B + \log\omega_{BC}))\notag\\
&&\geqslant -2\log\Tr{\sqrt{\rho_{ABC}}\sqrt{\exp(\log\sigma_{AB} - \log\tau_B + \log\omega_{BC})}}\\
&&\geqslant\norm{\sqrt{\rho_{ABC}} - \sqrt{\exp(\log\sigma_{AB} - \log\tau_B + \log\omega_{BC})}}^2_2\\
&&\geqslant \frac14\norm{\rho_{ABC} - \exp(\log\sigma_{AB} -
\log\tau_B + \log\omega_{BC})}^2_1.
\end{eqnarray}
\end{thrm}

\begin{proof}
Since $\Tr{\exp(\log\sigma_{AB} - \log\tau_B +
\log\omega_{BC})}\leqslant1$ (see Remark~\ref{remark-1}), that is
$\exp(\log\sigma_{AB} - \log\tau_B + \log\omega_{BC})$ is a
subnormalized state, it follows from \eqref{eq:sub-normalized-state}
that the desired inequality is correct.
\end{proof}

\begin{prop}
For a tripartite state $\rho_{ABC}\in\density{\cH_{ABC}}$, it holds
that
\begin{eqnarray}
\rS(\rho_{AB}) + \rS(\rho_{BC}) - \rS(\rho_{ABC}) &\geqslant&
-2\log\Tr{\sqrt{\rho_{ABC}}\sqrt{\exp(\log\rho_{AB}+\log\rho_{BC})}}\\
&\geqslant&\norm{\sqrt{\rho_{ABC}} - \sqrt{\exp(\log\rho_{AB} + \log\rho_{BC})}}^2_2\\
&&\geqslant \frac14\norm{\rho_{ABC} - \exp(\log\rho_{AB} +
\log\rho_{BC})}^2_1.
\end{eqnarray}
\end{prop}

\begin{proof}
All we need to do is to rewrite $\rS(\rho_{AB}) + \rS(\rho_{BC}) -
\rS(\rho_{ABC})$ as a relative entropy:
$$
\rS(\rho_{AB}) + \rS(\rho_{BC}) - \rS(\rho_{ABC}) =
\rS(\rho_{ABC}||\exp(\log\rho_{AB}+\log\rho_{BC})).
$$
Next, we prove that $\exp(\log\rho_{AB}+\log\rho_{BC})$ is a
subnormalized state. Using Golden-Thompson inequality, we have
\begin{eqnarray}
\Tr{\exp(\log\rho_{AB}+\log\rho_{BC})}&\leqslant&
\Tr{\exp(\log\rho_{AB})\exp(\log\rho_{BC})}\\
&\leqslant& \Tr{\rho_{AB}\rho_{BC}} = \Tr{\rho^2_B}\leqslant 1.
\end{eqnarray}
This completes the proof.
\end{proof}

Further comparison with the inequalities in
\cite{Kim2012,Ruskai2012} would be interesting and is left for the
future research.

%===============================================================%
\section{Discussion and concluding remarks}\label{sect:Remarks}
%===============================================================%

The lower bound in \eqref{eq:conditional-lower-bound} is clearly
independent of any measurement, compared with \eqref{eq:Fernando}
and \eqref{eq:keli}. Since the trace-norm decreases under generic
quantum channels, in particular under partial trace, it follows that
\begin{eqnarray}
E_{sq}(\rho_{AC}) \geqslant \frac18 \norm{\rho_{AC} -
\Ptr{B}{\exp\Pa{\log\rho_{AB} + \log\rho_{BC} -\log\rho_B}}}^2_1,
\end{eqnarray}
where $E_{sq}$ is an entanglement measure, i.e. \emph{squashed
entanglement}, defined by
\begin{eqnarray}
E_{sq}(\rho_{AC}) = \inf\Set{\frac12I(A:C|B)_\rho:
\ptr{B}{\rho_{ABC}}=\rho_{AC}},
\end{eqnarray}
where the infimum is taken over all possible extensions $\rho_{ABC}$
of $\rho_{AC}$. Since the squashed entanglement is a \emph{faithful}
measure, i.e. $E_{sq}(\rho_{AC}) = 0$ if and only if $\rho_{AC}$ is
a separable state, it follows that if $\rho_{AC}$ is separable, then
there exists an extension $\rho_{ABC}$ of $\rho_{AC}$ such that
\begin{eqnarray}
\rho_{AC} = \Ptr{B}{\exp\Pa{\log\rho_{AB} + \log\rho_{BC}
-\log\rho_B}}.
\end{eqnarray}
Equivalently, if $\rho_{AC}\neq\Ptr{B}{\exp\Pa{\log\rho_{AB} +
\log\rho_{BC} -\log\rho_B}}$ for any extension $\rho_{ABC}$ of
$\rho_{AC}$, then $\rho_{AC}$ must be entangled. We would like to
know wether or not
\begin{eqnarray}
\rho_{AC}\quad\text{is separable if and only if}\quad
\rho_{AC}=\ptr{B}{\exp\Pa{\log\rho_{AB} + \log\rho_{BC}-\log\rho_B}}
\end{eqnarray}
for some extension $\rho_{ABC}$ of $\rho_{AC}$. From this
observation, one sees that finding some properties of the following
operators is a very interesting subject:
$$
\exp\Pa{\log\rho_{AB} + \log\rho_{BC}
-\log\rho_B},~~\Ptr{B}{\exp\Pa{\log\rho_{AB} + \log\rho_{BC}
-\log\rho_B}}.
$$
Taking the partial traces of tripartite operators maybe important in
the investigation of entanglement theory. It will provide new
insights in understanding entanglement.

For instance, consider the \emph{generalized Lie-Trotter product
formula} \cite{Suzuki}: for any $k$ matrices $A_1,\ldots,A_k$, it
holds that
$$
\lim_{n\to\infty}\Pa{\exp(A_1/n)\exp(A_2/n)\cdots\exp(A_k/n)}^n =
\exp(A_1+A_2+\cdots+A_k).
$$
This leads to the following identities:
\begin{eqnarray}
\exp(\log\rho_{AB} - \log\rho_B + \log\rho_{BC}) &=&
\lim_{n\to\infty}
\Pa{\rho^{1/2n}_{AB}\rho^{-1/2n}_B\rho^{1/n}_{BC}\rho^{-1/2n}_B\rho^{1/2n}_{AB}}^n\\
&=&\lim_{n\to\infty}
\Pa{\rho^{1/2n}_{BC}\rho^{-1/2n}_B\rho^{1/n}_{AB}\rho^{-1/2n}_B\rho^{1/2n}_{BC}}^n.
\end{eqnarray}
We wonder whether the sequence
$\Tr{\Pa{\rho^{1/2n}_{AB}\rho^{-1/2n}_B\rho^{1/n}_{BC}\rho^{-1/2n}_B\rho^{1/2n}_{AB}}^n}$
is \emph{monotone} in $n$ and is no more than one. I proposed this
conjecture in the previous version of the present paper, and
luckily, Datta and Wilde \cite{DW} gave a positive answer to my
question partially. In fact, they found that a key result should be
cited. That is, if $M\in M_d(\complex)$ and $\alpha\geqslant1$, then
the map $X\mapsto \Tr{\Br{MX^{1/\alpha}M^\dagger}^\alpha}$ is
concave over $\pd{\complex^d}$ \cite{CL}. In what follows, we may
obtain the following:
\begin{eqnarray*}
1&=&\Tr{\Pa{\rho^{1/2n}_{AB}\rho^{-1/2n}_B\Br{\rho_B\ot\I_C/d_C}^{1/n}\rho^{-1/2n}_B\rho^{1/2n}_{AB}}^n}\\
&=&\Tr{\Pa{\rho^{1/2n}_{AB}\rho^{-1/2n}_B\Br{\int_{\unitary{d_C}}
U_C\rho_{BC}U^\dagger_C
\dif U_C}^{1/n}\rho^{-1/2n}_B\rho^{1/2n}_{AB}}^n}\\
&\geqslant&\int_{\unitary{d_C}}\Tr{\Pa{\rho^{1/2n}_{AB}\rho^{-1/2n}_B\Br{
U_C\rho_{BC}U^\dagger_C
}^{1/n}\rho^{-1/2n}_B\rho^{1/2n}_{AB}}^n}\dif U_C\\
&=&\Tr{\Pa{\rho^{1/2n}_{AB}\rho^{-1/2n}_B\rho^{1/n}_{BC}\rho^{-1/2n}_B\rho^{1/2n}_{AB}}^n}.
\end{eqnarray*}
Note here that the concavity of the map $X\mapsto
\Tr{\Br{MX^{1/\alpha}M^\dagger}^\alpha}$ is very essential here. The
following inequality
\begin{eqnarray}\label{eq:tr-ineq}
\Tr{\Pa{\rho^{1/2n}_{AB}\rho^{-1/2n}_B\rho^{1/n}_{BC}\rho^{-1/2n}_B\rho^{1/2n}_{AB}}^n}\leqslant
1~~(\forall n\geqslant1)
\end{eqnarray}
directly leads to another proof of the fact that
$\Tr{\exp(\log\rho_{AB} - \log\rho_B + \log\rho_{BC})}\leqslant1$ by
taking the limit of Eq.~\eqref{eq:tr-ineq} when $n\to\infty$. We can
use the same technique to get a more general result: If
$\rho_{ABC},\sigma_{ABC},\tau_{ABC}\in\density{\cH_{ABC}}$ and
$\rho_B=\sigma_B$ or $\sigma_B=\tau_B$, then
\begin{eqnarray}
\Tr{\Pa{\rho^{1/2n}_{AB}\sigma^{-1/2n}_B\tau^{1/n}_{BC}\sigma^{-1/2n}_B\rho^{1/2n}_{AB}}^n}\leqslant
1~~(\forall n\geqslant1).
\end{eqnarray}

Based on this result and Stinespring's dilation representation of
completely positive maps, Datta and Wilde \cite[Eq.~(3.50)]{DW} gave
the proof of the following inequality:
\begin{eqnarray}
\Tr{\exp\Pa{\log\sigma + \Phi^*(\log\Phi(\rho)) -
\Phi^*(\log\Phi(\sigma))}}\leqslant1.
\end{eqnarray}
It is natural to ask here the implication of the saturation:
\begin{eqnarray}
\Tr{\exp\Pa{\log\sigma +
\Phi^*(\log\Phi(\rho)) - \Phi^*(\log\Phi(\sigma))}}=1.
\end{eqnarray}
It is left \emph{open} for the future research. Furthermore, we get
the following improvement of monotonicity of relative entropy:
\begin{eqnarray}
\rS(\rho||\sigma) - \rS(\Phi(\rho)||\Phi(\sigma))\geqslant
\frac14\norm{\rho - \exp\Pa{\log\sigma + \Phi^*(\log\Phi(\rho)) -
\Phi^*(\log\Phi(\sigma))}}^2_1.
\end{eqnarray}
Clearly the above result is not applicable in the present form since
the operator
$$
\exp\Pa{\log\sigma + \Phi^*(\log\Phi(\rho)) -
\Phi^*(\log\Phi(\sigma))}
$$
may not be a valid state. We propose the
following \emph{open} questions:
\begin{eqnarray}
\rS(\rho||\sigma) - \rS(\Phi(\rho)||\Phi(\sigma))&\geqslant&
\frac14\norm{\rho - \Phi^*_\sigma\circ\Phi(\rho)}^2_1,\label{eq:stronger-mono}\\
\rS(\rho_{AB}||\sigma_{AB}) - \rS(\rho_A||\sigma_A) &\geqslant&
\frac14\norm{\rho_{AB} - \sigma^{1/2}_{AB}\sigma^{-1/2}_A\rho_A\sigma^{-1/2}_A\sigma^{1/2}_{AB}}^2_1,\\
I(A:C|B)_\rho &\geqslant&\frac14 \norm{\rho_{ABC} -
\rho^{1/2}_{AB}\rho^{-1/2}_B\rho_{BC}\rho^{-1/2}_B\rho^{1/2}_{AB}}^2_1.
\end{eqnarray}
Apparently, \eqref{eq:stronger-mono} implies the remaining two
inequalities. Further investigations on these topics are in order.
We hope that the results obtained in our work shed new light over
related subjects in quantum information theory.

%=============================================================================%

\subsection*{Acknowledgements}
The author thanks the anonymous referee for helpful comments on our
manuscript, and he also thanks Uttam Singh and Zhaoqi Wu for various
helpful discussions and suggestions. The work is supported by
National Natural Science Foundation of China (No.11301124).

%===========================================================================%

\end{document}